\documentclass[11pt]{amsart}
\usepackage{geometry}                
\usepackage{graphicx}
\usepackage{amssymb}
\usepackage{epstopdf}
\usepackage{enumerate}
\newcommand{\Q}{\mathbb{Q}}
\newcommand{\N}{\mathbb{N}}
\newcommand{\Z}{\mathbb{Z}}

\newcommand{\sgn}[1]{\mathrm{sgn}(#1)}

\newcommand{\tab}{\;\;\;\;\;}

\newtheorem{lemma}{Lemma}
\newtheorem{theorem}{Theorem}

\theoremstyle{definition}
\newtheorem*{definition}{Definition}

\theoremstyle{remark}
\newtheorem*{remark}{Remark}

\title{The reachability problem for affine functions on the integers}
\author{Daniel Fremont}
\thanks{The author wishes to thank Henry Cohn, without whose copious and very helpful advice this work would not have been possible. Thanks are also due to the many people who gave feedback and suggestions, in particular Christoph Haase.}
\address{Department of Mathematics \\
Massachusetts Institute of Technology \\
Cambridge, MA 02139}
\email{dfremont@mit.edu}
\date{December 15, 2012}
\keywords{reachability problems; affine functions; regular expressions}

\begin{document}

\begin{abstract}
We consider the problem of determining, given $x, y \in \Z^k$ and a finite set $F$ of affine functions on $\Z^k$, whether $y$ is reachable from $x$ by applying the functions $F$. We also consider the analogous problem over $\N^k$. These problems are known to be undecidable  for $k \ge 2$. We give \textsf{2-EXPTIME} algorithms for both problems in the remaining case $k=1$. The exact complexities remain open, although we show a simple \textsf{NP} lower bound.
\end{abstract}

\maketitle

\section{Introduction}

Many dynamical systems with simple evolution rules nevertheless exhibit unpredictable long-term behavior. Multidimensional systems in particular can easily be so complex that questions like state reachability are undecidable. For instance, this is the case for states in $[0,1]^2$ evolving under a piecewise-linear function \cite{koiran}. There are even simpler nondeterministic examples, such as states in $\Q^2$ under a finite set of affine functions \cite{bell-potapov}. For both systems, having states with two coordinates whose evolution is not independent is essential to the undecidability proofs. Generally, while it is often not difficult to prove undecidability for systems with sufficiently high dimension, determining if and when the transition to decidability occurs at lower dimensions is harder. In particular, it is not known whether the reachability problem is decidable for nondeterministic affine evolution on $\Q$.

In this paper, we consider the simpler problem of reachability under nondeterministic affine evolution on $\Z$: given $x, y \in \Z$ and a finite set $F$ of affine functions $f_i(z) = a_i z + b_i$ with $a_i, b_i \in \Z$, determine whether $y$ is reachable from $x$ by applying functions in $F$. The generalizations to $\Z^n$ are undecidable for all $n \ge 2$, as implicitly shown in \cite{paterson} (and a little more clearly in \cite[Section 4.9]{gaubert-katz}). We prove that the remaining case, $n = 1$, is decidable, giving a \textsf{2-EXPTIME} algorithm for it in Section \ref{reach-z}. We also consider the version of this problem with evolution over $\N$: in this problem, the functions $f_i$ can still have negative coefficients, but may not be applied if they would yield a negative result. In Section \ref{reach-n} we give a \textsf{2-EXPTIME} algorithm for this problem by modifying our algorithm for the case over $\Z$. Finally, in Section \ref{lower-bound} we show that the problems over $\Z$ and $\N$ are both \textsf{NP}-hard.

\section{Affine reachability over $\Z$} \label{reach-z}

We begin by defining some notation that we will use throughout this paper.
\begin{definition}
If $S$ is a set of affine functions on $\Z$, we will call any function of the form $G = s_K \circ \dots \circ s_1$ for some $K \in \N$ with each $s_i \in S$ an $S$-\emph{composition}. We will want to discuss the individual functions $s_i$ which appear in an $S$-composition, so $G$ is formally the tuple $(s_1, \dots, s_K)$, but we will often view $S$-compositions as functions without further comment. The \emph{orbit} of $G$ applied to the argument $x$ is the set of values $\{x, s_1(x), (s_2 \circ s_1)(x), \dots, (s_K \circ \dots \circ s_1)(x) \}$. We write $x \xrightarrow{S} y$ \emph{via} $G$ to indicate that $G$ is an $S$-composition such that $G(x) = y$, and $x \xrightarrow{S} y$ to assert the existence of such a $G$.
\end{definition}
In this notation, the affine reachability problem over $\Z$ is to determine, given $x, y \in \Z$ and a finite set $F$ of affine functions $f_i(z) = a_i z + b_i$ with $a_i, b_i \in \Z$, whether $x \xrightarrow{F} y$. The problem takes several qualitatively different forms depending on the values of the linear coefficients $a_i$. The simplest nontrivial case is when they all satisfy $|a_i| > 1$. 

\begin{lemma} \label{expanding-reachable}
There is an \textsf{EXPTIME} algorithm to decide, given any $x, y \in \Z$ and a finite set $F$ of functions $f_i(z) = a_i z + b_i$ with $a_i, b_i \in \Z$ and satisfying $|a_i| > 1$, whether $x \xrightarrow{F} y$.
\end{lemma}
\begin{proof}
Outside of some finite interval, for instance $[-Q,Q]$ with $Q = 1 + \max |b_i|$, each function $f_i$ strictly increases absolute value. Putting $R = \max \{Q, |y|\}$, for any $F$-composition $G$ and $z \in \Z$ with $|z| > R$ we have $|G(z)| > |z| > |y|$ and thus $G(z) \ne y$. This means that all preimages of $y$ under $F$-compositions must lie in the finite interval $I = [-R,R]$. Create a directed graph $D$ with a vertex for each integer $z \in I$, and add edges from $z$ to $f_i(z)$ for each $i$ satisfying $f_i(z) \in I$. Since every preimage of $y$ under an $F$-composition lies in $I$, we have $x \xrightarrow{F} y$ if and only if $x \in I$ and there is a path in $D$ from $x$ to $y$. We can determine whether such a path exists in exponential time using graph search, since $D$ has exponentially-many vertices (at most linearly-many in the values of $b_i$ and $y$).
\end{proof}
\begin{remark}
As will be important later, with a small modification of this algorithm we can handle the presence of one $f_j$ with $a_j = -1$, so that $f_j(z) = -z+b_j$. The only change necessary is to broaden $I$ to the interval $I' = [\min(-R,-R+b_j), \max(R,R+b_j)]$. Then $f_j^{-1} = f_j$ maps $I'$ onto itself, so all preimages of $y$ under $F$-compositions are in $I'$, and the argument goes through as above.
\end{remark}

However, when a function in $F$ is of the form $g(z) = z + k$ this method breaks down, because the preimages of $y$ under $F$-compositions are no longer bounded. This also happens if there are two functions of the form $f(z)=-z+b$, since then their composition is of the form $g(z)=z+k$. Fortunately, functions like $g$ can contribute to an $F$-composition in basically only one way.
\begin{lemma}
\label{extraction-lemma}
For any set $S$ of functions $f_i(z) = a_i z + b_i$ with $a_i, b_i \in \Z$ for $1 \le i \le N$ and function $f_0(z) = z + k$, put $F = S \cup \{f_0\}$. Then for any $F$-composition $G$, we have:
\begin{enumerate}[\tab (a)]
\item If $G = f_{e_0} \circ \dots \circ f_{e_j} \circ f_0^n \circ f_{e_{j+1}} \circ \dots \circ f_{e_K}$, then $G(z) = H(z) + ank$ where $H = f_{e_0} \circ \dots \circ f_{e_K}$ and $a = a_{e_0} \dots a_{e_j}$. \label{single-extraction}
\item $G(z) = H(z) + ak$ for some $S$-composition $H$ and $a \in \Z$. If $a_i > 0$ for every function $f_i$ appearing in $G$, then $a \ge 0$. \label{total-extraction}
\end{enumerate}
\end{lemma}
\begin{proof}
\begin{enumerate}[(a)]
\item We have $(f_i \circ f_0^n)(z) = a_i (z + nk) + b_i = (a_i z + b_i) + a_i n k = (f_0^{a_i n} \circ f_i)(z)$, so $G = f_{e_0} \circ \dots \circ f_{e_j} \circ f_0^n \circ f_{e_{j+1}} \circ \dots \circ f_{e_K} = f_{e_0} \circ \dots \circ f_{e_{j-1}} \circ f_0^{a_{e_j} n} \circ f_{e_j} \circ \dots \circ f_{e_K}$. Repeating this $j$ more times gives $G = f_0^{a_{e_0} \dots a_{e_j} n} \circ f_{e_0} \circ \dots \circ f_{e_K} = H(z) + ank$ with $H = f_{e_0} \circ \dots \circ f_{e_K}$ and $a = a_{e_0} \dots a_{e_j}$.
\item Apply the previous result once for each instance of $f_0$ in $G$, giving $G = f_0^{c_0} \circ \dots \circ f_0^{c_L} \circ f_{e_0} \circ \dots \circ f_{e_K}$ for some $c_0, \dots c_L \in \Z$ which are products of the coefficients $a_i$. Then $G(z) = H(z) + ak$ with $H$ the $S$-composition $H = f_{e_0} \circ \dots \circ f_{e_K}$ and $a = c_0 + \dots + c_L$. If $a_i > 0$ for each $f_i$ appearing in $G$, then $c_i > 0$ and so $a \ge 0$ (we could have $a = 0$ if there were no instances of $f_0$ in $G$). \qedhere
\end{enumerate}
\end{proof}
We now have several cases, based on which $S$-compositions $G$ satisfy $x \xrightarrow{S} y \pmod k$ via $G$.
\begin{lemma} \label{lemma-cases}
For any $x,y \in \Z$, finite set $S$ of functions $f_i(z) = a_i z + b_i$ with $a_i, b_i \in \Z$ and $a_i \ne 0$, and function $g(z) = z + k$ with $k \in \Z$ and $k \ne 0$, put $F = S \cup \{g\}$ and $\mathbf{G} = \{ G : x \xrightarrow{S} y \pmod k \: \text{via} \; G \}$. Then the following are true:
\begin{enumerate}[\tab (A)]
\item If $\mathbf{G} = \emptyset$, then $x \not \xrightarrow{F} y$. \label{case-none}
\item If some $G \in \mathbf{G}$ satisfies $\sgn{G(x) - y} \ne \sgn{k}$, then $x \xrightarrow{F} y$. \label{case-opposite-sign}
\item If some $G \in \mathbf{G}$ with $G = f_{e_0} \circ \cdots \circ f_{e_K}$ has $a_{e_j} < 0$ for some $j$, then $x \xrightarrow{F} y$. \label{case-negative-coeff}
\item If none of the above cases hold, then $x \not \xrightarrow{F} y$. \label{case-same-sign}
\end{enumerate}
\end{lemma}
\begin{proof}
\begin{enumerate}[(A)]
\item By Lemma \ref{extraction-lemma}\ref{total-extraction}, any $F$-composition can be written as an $S$-composition plus a multiple of $k$. If $\mathbf{G} = \emptyset$, then no $S$-composition can reach $y \pmod k$ from $x$, and therefore neither can any $F$-composition.
\item If some $G \in \mathbf{G}$ satisfies $\sgn{G(x) - y} \ne \sgn{k}$, then either $\sgn{G(x) - y} = 0$ or $\sgn{G(x) - y} = -\sgn{k}$. If the first of these is true, then $G(x) = y$ and so $x \xrightarrow{F} y$ via $G$. Otherwise, there is some $n \in \N$ such that $G(x) - y = -nk$, so $y = G(x) + nk = (g^n \circ G)(x)$. Putting $G' = g^n \circ G$, we have $x \xrightarrow{F} y$ via $G'$.
\item If some $G \in \mathbf{G}$ with $G = f_{e_0} \circ \cdots \circ f_{e_K}$ has $a_{e_j} < 0$ for some $j$, take the smallest such $j$. Defining $G' = f_{e_0} \circ \cdots f_{e_j} \circ g^n \circ f_{e_{j+1}} \circ \cdots \circ f_{e_K}$, by Lemma \ref{extraction-lemma}\ref{single-extraction} we have $G'(z) = G(z) + a n k$ with $a = a_{e_0} a_{e_1} \cdots a_{e_j}$, and because $a_{e_k} > 0$ for all $k < j$ by our choice of $j$, we have $a < 0$. We may assume that case \ref{case-opposite-sign} does not hold, since otherwise we have $x \xrightarrow{F} y$ immediately as shown above. Then we have $\sgn{G(x) - y} = \sgn{k}$, and so $\sgn{G(x) - y} = -\sgn{ak}$. Therefore with $n$ sufficiently large we have $\sgn{G'(x) - y} = \sgn{G(x) - y + ank} = -\sgn{G(x) - y} = -\sgn{k}$. Then as in case \ref{case-opposite-sign}, we have $x \xrightarrow{F} y$ via $G'' = g^m \circ G'$ for some $m \in \N$.
\item Suppose that $x \xrightarrow{F} y$ via some $H$. Then by Lemma \ref{extraction-lemma}\ref{total-extraction} we have $H = g^a \circ G$ for some $S$-composition $G$ and $a \in \Z$. Now $G(x) = H(x) - ak \equiv H(x) = y \pmod k$, so $G \in \mathbf{G}$. Since case \ref{case-opposite-sign} does not hold, we have $\sgn{G(x) - y} = \sgn{k}$. Since case \ref{case-negative-coeff} does not hold, we have $a \ge 0$, again by Lemma \ref{extraction-lemma}\ref{total-extraction}. If $a = 0$, then $G(x) = H(x) = y$, so $\sgn{G(x) - y} = 0 \ne \sgn{k}$ and case \ref{case-opposite-sign} holds, contrary to our assumption. So $a > 0$, and thus $\sgn{G(x) -y} = \sgn{k} = \sgn{ak}$. But this is impossible, since $G(x) - y = H(x) - y - ak = -ak$. So we cannot have $x \xrightarrow{F} y$. \qedhere
\end{enumerate}
\end{proof}

To test cases (\ref{case-none}), (\ref{case-opposite-sign}), and (\ref{case-negative-coeff}), we use the following algorithm.

\begin{lemma} \label{largest-mod-reachable}
Given any $x, y, k \in \Z$ with $k \ne 0$ and a set $F$ of functions $f_i(z) = a_i z + b_i$ with $a_i, b_i \in \Z$ and $a_i \ne 0$ for $1 \le i \le N$, put $\mathbf{G} = \{ G : x \xrightarrow{F} y \pmod k \: \text{via} \; G \}$. There is a \textsf{2-EXPTIME} algorithm such that:
\begin{enumerate}
\item If $\mathbf{G} = \emptyset$, the algorithm returns \textsc{Empty}.
\item If there is some $G \in \mathbf{G}$ with $G = f_{e_0} \circ \cdots \circ f_{e_K}$ and $a_{e_j} < 0$ for some $j$, the algorithm returns \textsc{Negative}.
\item Otherwise, the algorithm returns $\sup \; \{ G(x) : G \in \mathbf{G} \}$.
\end{enumerate}
\end{lemma}

The flags \textsc{Empty} and \textsc{Negative} are enough for us to detect cases \ref{case-none} and \ref{case-negative-coeff} of Lemma \ref{lemma-cases}. If $k < 0$, the value $V = \sup \; \{ G(x) : G \in \mathbf{G} \}$ allows us to recognize case \ref{case-opposite-sign}, because this case holds if and only if $V \ge y$. If $k > 0$ we need the value $V' = \inf \; \{ G(x) : G \in \mathbf{G} \}$ instead, since then case \ref{case-opposite-sign} holds if and only if $V' \le y$. The modifications to the algorithm of Lemma \ref{largest-mod-reachable} required to make it compute $V'$ instead of $V$ are simple and obvious (just exchanging ``increases'' with ``decreases'' in several places, etc.), so we omit them. Now we prove Lemma \ref{largest-mod-reachable}, assuming a couple of auxiliary lemmas (Lemmas \ref{I-lemma} and \ref{I-algorithm}) which we will return to afterwards.

\begin{proof}[Proof of Lemma \ref{largest-mod-reachable}]
First we check if there is \emph{any} $F$-composition mapping $x$ to $y \pmod k$. Create a directed graph $D$ with a vertex for each congruence class mod $k$. Add edges indicating which classes are mapped to which under each $f_i$. Then there is a path in $D$ from the congruence class of $x$ to the congruence class of $y$ if and only if $x \xrightarrow{F} y \pmod k$. Use graph search to determine if there is such a path, and return \textsc{Empty} if not. Since $D$ has $|k|$ vertices, this search takes exponential time.

If there are paths from $x \pmod k$ to $y \pmod k$, we need to analyze all of them to see which ones yield the largest final value. We can conveniently describe the paths using regular expressions. Consider $D$ to be a deterministic finite automaton, where an input symbol $e \in \{1, \dots, N\}$ causes the edge corresponding to applying $f_e$ to be followed. Let the initial state be $x \pmod k$, and the only accepting state be $y \pmod k$. If $s = e_1 \dots e_K$ is a sequence of input symbols, we write $P_s = f_{e_K} \circ \dots \circ f_{e_1}$ (note the order!), and then $D$ accepts $s$ if and only if $P_s(x) \equiv y \pmod k$.

Now we convert $D$ into a regular expression $R$ with the same language $L(R)$. We write concatenation multiplicatively, use $|$ to denote union/alternation, and use $\epsilon$ and $\emptyset$ as the symbols for the empty string and the empty language respectively. Because $D$ has exponentially-many vertices (and a linearly-sized input alphabet), $R$ has at most doubly-exponential size $|R|$, and the conversion from $D$ to $R$ can be done in time at most polynomial in $|R|$ (see \cite{mcnaughton-yamada, dfa-to-re-upper}). We store $R$ as a tree, with literals, $\epsilon$, or $\emptyset$ at the leaf nodes and operators at the other nodes. The ``length'' $|R|$ in this representation is just the total number of nodes.

Next, reduce $R$ to not include the symbol $\emptyset$ by repeatedly passing through $R$ applying the identities $E | \emptyset = E$, $E\emptyset = \emptyset$, and $\emptyset^* = \epsilon$ for any expression $E$. Each pass takes time linear in the length of $R$ and strictly decreases its length, so there can be at most $|R|$ passes and the total time taken is $O(|R|^2)$. Afterwards, if the symbol $\emptyset$ appears in $R$ it must not be operated on by any operator, since otherwise one of the identities above would apply. Therefore $\emptyset$ can appear in $R$ only if $R = \emptyset$, but since $L(R)$ is nonempty (because we returned \textsc{Empty} above if so) this is not the case. So $R$ does not contain the symbol $\emptyset$.

Now if $R$ contains a literal corresponding to a function $f_i$ with $a_i < 0$ (which can obviously be determined in $O(|R|)$ time), then $x \xrightarrow{F} y \pmod k$ via an $F$-composition which includes $f_i$, so we return \textsc{Negative}. Otherwise, we convert $R$ into disjunctive normal form $S_1 | S_2 | \cdots | S_M$ where each $S_i$ has no union operations, by iteratively applying the identities $(\alpha | \beta)^* = (\alpha^* \beta^*)^*$, $\alpha (\beta | \gamma) = \alpha \beta | \alpha \gamma$, and $(\alpha | \beta) \gamma = \alpha \gamma | \beta \gamma$. Each identity either decreases the number of unions or moves one closer to the topmost level, so this process will also take time polynomial in $|R|$.

We say that a regular expression $E$ is \emph{reduced} if it contains only literals appearing in $R$ and has no $\emptyset$ symbols or union operations. The reductions we have done above ensure that any expression produced by concatenating subexpressions of the clauses $S_i$ is reduced. Since every literal in $R$ corresponds to a function $f_i$ with $a_i > 0$ (because we would have returned \textsc{Negative} otherwise), and the composition of two linear polynomials with positive linear coefficients has a positive linear coefficient, for any reduced expression $E$ and any $s \in L(E)$, $P_s$ has a positive linear coefficient --- this will be important in a moment. We will want to refer to those $F$-compositions which are generated by reduced expressions, and to decrease the proliferation of symbols, we will say that an $F$-composition $G$ \emph{matches} the expression $E$ if there is some $s \in L(E)$ such that $G = P_s$. Then $\mathbf{G}$ consists precisely of those $F$-compositions which match $R$.

Now, given some $z \in \Z$ and a reduced expression $E$, define $I(z,E)$ to mean that $\exists s \in L(E) : P_s(z) > z$. In words, $I(z,E)$ is true if and only if there is some $F$-composition matching $E$ which increases $z$. If $L(E)$ is finite, computing $I(z,E)$ is only a matter of testing various cases --- the difficulty is handling expressions with stars. Fortunately, we can reduce $I$ to its values on expressions with fewer stars using the identity $I(z, \ell \alpha^* \beta) \iff I(P_\ell(z), \alpha) \lor I(z, \ell \beta)$, which we will prove in Lemma \ref{I-lemma}. This then allows us to compute $I(z,E)$ recursively in polynomial time, as we will show in Lemma \ref{I-algorithm}.

Now we are ready to return to the main problem. For each clause $S_i$, we want to find the supremum $V_i$ of the possible values $x$ is mapped to by any $F$-composition matching $S_i$. To do this, we keep track of the supremum of the values $x$ is mapped to by $F$-compositions which match progressively-longer prefixes of $S_i$. Write $S_i = T_1 \dots T_K$ by flattening out concatenations, so that each $T_j$ is either a literal or a starred subexpression. Let $x^{(i)}_j = \sup \; \{ P_s(x) : s \in L(T_1 \dots T_j) \}$ for $1 \le j \le K$. Clearly $V_i = x^{(i)}_K$, and we put $x^{(i)}_0 = x$ (since the largest possible value reachable after applying no functions is the starting value $x$). For $j \ge 1$, we calculate $x^{(i)}_j$ in terms of $x^{(i)}_{j-1}$ as follows. If $T_j$ is a literal, then any $F$-composition matching $T_1 \dots T_j$ must be of the form $P_{T_j} \circ q$ where $q$ is an $F$-composition matching $T_1 \dots T_{j-1}$. Since by definition the largest possible value of $q(x)$ is $x^{(i)}_{j-1}$, and $P_{T_j}$ has a positive linear coefficient, the largest possible value of $(P_{T_j} \circ q)(x)$ is $P_{T_j}(x^{(i)}_{j-1})$. Thus $x^{(i)}_j = P_{T_j}(x^{(i)}_{j-1})$. If $T_j$ is a starred subexpression instead, $T_j = \alpha^*$, we compute $I(x^{(i)}_{j-1}, \alpha)$. If this is true, then some $F$-composition $p$ matching $\alpha$ increases $x^{(i)}_{j-1}$, and because $p$ has a positive linear coefficient it must increase all values larger than $x^{(i)}_{j-1}$. So we can increase $x^{(i)}_{j-1}$ as much as we want by repeatedly applying $p$, and thus $x^{(i)}_j = \infty$. If $I(x^{(i)}_{j-1}, \alpha)$ is false, then no $F$-composition matching $\alpha$ increases $x^{(i)}_{j-1}$, so $x^{(i)}_j = x^{(i)}_{j-1}$ (since $\alpha^*$ is matched by $P_\epsilon$, which leaves $x^{(i)}_{j-1}$ fixed). Thus we can iteratively compute $V_i$, beginning with $x^{(i)}_0 = x$ and proceeding through $x^{(i)}_K = V_i$. There are $O(|S_i|)$ intermediate values $x^{(i)}_j$ which need to be computed, and each one requires at most one call to $I$ on an expression of size at most $O(|S_i|)$. Thus we can compute each $V_i$ in time polynomial in $|S_i|$, and all the values $V_i$ together in time polynomial in $|R|$.

Now put $V = \max V_i$. Because the union of the languages of each clause $S_i$ is the language of $R$, $V$ is the largest value reachable using $F$-compositions matching $R$ (or $\infty$ if arbitrarily large values are reachable). Therefore $V = \sup \; \{ G(x) : G \in \mathbf{G} \}$, the desired value, and we return it.

As mentioned above, the first stage of this algorithm takes exponential time, and all subsequent stages take time at most polynomial in $|R|$. Since $|R|$ is at most doubly-exponential in the input length, the entire algorithm takes at most doubly-exponential time.
\end{proof}

Lemma \ref{largest-mod-reachable} depends on our ability to efficiently calculate $I(z,E)$. As was mentioned above, the key to doing this is to reduce $I(z,E)$ to values of $I$ on smaller subexpressions, as made possible by the following lemma.

\begin{lemma} \label{I-lemma}
If $z \in \Z$, $\alpha$ and $\beta$ are reduced regular expressions, and $\ell$ is a sequence of literals, the following are true:
\begin{enumerate}
\item $I(z, \ell \alpha^* \beta) \iff I(z, \ell \beta) \lor I(P_\ell(z), \alpha)$ \label{I-lemma-literal-star-prefix}
\item $I(z, \ell \alpha^*) \iff I(z, \ell) \lor I(P_\ell(z), \alpha)$ \label{I-lemma-literal-star}
\item $I(z, \alpha^* \beta) \iff I(z, \alpha) \lor I(z, \beta)$ \label{I-lemma-star-prefix}
\item $I(z, \alpha^*) \iff I(z,\alpha)$ \label{I-lemma-star}
\end{enumerate}
\end{lemma}
\begin{proof}
\begin{enumerate}
\item $I(z, \ell \beta)$ implies $I(z, \ell \alpha^* \beta)$ because $\ell \beta$ matches $\ell \alpha^* \beta$. If $I(P_\ell(z), \alpha)$ holds, then there is an $F$-composition $p$ matching $\alpha$ which increases $P_\ell(z)$. Because $p$ has a positive linear coefficient as observed above (since $\alpha$ is reduced), $p$ must increase anything greater than $P_\ell(z)$, and thus repeated applications of $p$ can increase $P_\ell(z)$ as much as desired. Therefore for any particular $F$-composition $q$ matching $\beta$, there is some $n \in \N$ such that $q \circ p^n \circ P_\ell$ increases $z$ (since $\beta$ is reduced and so $q$ also has a positive linear coefficient). Since $q \circ p^n \circ P_\ell$ matches $\ell \alpha^* \beta$, $I(z, \ell \alpha^* \beta)$ holds.

Suppose neither $I(z, \ell \beta)$ nor $I(P_\ell(z), \alpha)$ hold. Any $F$-composition matching $\ell \alpha^* \beta$ is of the form $q \circ p \circ P_\ell$ where $q$ matches $\beta$ and $p$ is a composition of $F$-compositions matching $\alpha$. Since by our assumption no polynomials matching $\alpha$ increase $P_\ell(z)$, and these have positive linear coefficients, they cannot increase anything smaller than $P_\ell(z)$. Thus no composition of $F$-compositions matching $\alpha$ increases $P_\ell(z)$, and so $p$ does not increase $P_\ell(z)$. Now, since $q \circ P_\ell$ does not increase $z$ by assumption (since it matches $\ell \beta$), and $q$ has a positive linear coefficient, $q \circ p \circ P_\ell$ does not increase $z$ either. Thus no $F$-composition matching $\ell \alpha^* \beta$ increases $z$, and $I(z, \ell \alpha^* \beta)$ does not hold.
\item Put $\beta = \epsilon$ in (\ref{I-lemma-literal-star-prefix}) and use $L(\ell \alpha^* \epsilon) = L(\ell \alpha^*)$ and $L(\ell \epsilon) = L(\ell)$.
\item Put $\ell = \epsilon$ in (\ref{I-lemma-literal-star-prefix}) and use $L(\epsilon \beta) = L(\beta)$ and $P_\epsilon(z) = z$.
\item Put $\beta = \epsilon$ in (\ref{I-lemma-star-prefix}), use $L(\alpha^* \epsilon) = L(\alpha^*)$, and note that $I(z, \epsilon)$ is clearly false. \qedhere
\end{enumerate}
\end{proof}

These relationships allow us to give a straightforward recursive algorithm to compute $I$.

\begin{lemma} \label{I-algorithm}
There is a \textsf{P} algorithm to compute $I(z, E)$ for any $z \in \Z$ and reduced regular expression $E$.
\end{lemma}
\begin{proof}
First, note that if an expression $\alpha$ is a sequence of literals, then only $P_\alpha$ matches it, and we may determine $I(z,\alpha) = (P_\alpha(z) > z)$ directly by evaluating $P_\alpha(z)$. Now we compute $I(z,E)$ recursively, breaking into cases based on the topmost operator or symbol of $E$:
\begin{itemize}
\item $E = \epsilon$: $I(z,E)$ is clearly false.
\item $E$ is a literal: As noted above, $I(z,E) = (P_E(z) > z)$ may be directly calculated.
\item $E = F^*$: By part \ref{I-lemma-star} of Lemma \ref{I-lemma}, $I(z,E) = I(z,F^*) = I(z,F)$.
\item $E = F G$: By flattening as necessary, we may write $E = F_1 F_2 \cdots F_K$ for some $K \in \N$ with $K \ge 2$ and where none of the subexpressions $F_i$ are concatenations. If any of the subexpressions $F_i$ are $\epsilon$, we simply drop them and renumber appropriately --- this obviously leaves $L(E)$ fixed. Thus we may assume that each subexpression $F_i$ is either a starred subexpression or a literal. If each one is a literal, then $E$ is a sequence of literals and as noted above $I(z,E)$ can be computed directly. Otherwise, find the smallest $j$ such that $F_j = \alpha^*$ is a starred subexpression. There are several cases:
\begin{itemize}
\item $j = 1$: Then $E = \alpha^* \beta$ where $\beta = F_2 \cdots F_K$, so by part \ref{I-lemma-star-prefix} of Lemma \ref{I-lemma} we have $I(z,E)=I(z,\alpha^* \beta)=I(z,\alpha) \lor I(z,\beta)$.
\item $j = K$: Then $E = \ell \alpha^*$ where $\ell = F_1 \cdots F_{K-1}$ is a sequence of literals, so by part \ref{I-lemma-literal-star} of Lemma \ref{I-lemma} we have $I(z,E) = I(z,\ell \alpha^*) = I(z, \ell) \lor I(P_\ell(z),\alpha)$.
\item $1 < j < K$: Then $E = \ell \alpha^* \beta$ where $\ell = F_1 \cdots F_{j-1}$ is a sequence of literals and $\beta = F_{j+1} \cdots F_K$, so by part \ref{I-lemma-literal-star-prefix} of Lemma \ref{I-lemma} we have $I(z,E) = I(z,\ell \alpha^* \beta) = I(z,\ell \beta) \lor I(P_\ell(z),\alpha)$.
\end{itemize}
\end{itemize}

In each case $I(z,E)$ is either directly computable, equivalent to $I(x, F)$ for some $x \in \Z$ and $F$ a proper subexpression of $E$, or equivalent to $I(x,F) \lor I(y, G)$ for some $x,y \in \Z$, $F$ a proper subexpression or concatenation of disjoint subexpressions of $E$, and $G$ a subexpression of $E$ disjoint from $F$. Because $I(z,E)$ is always reduced to values of $I$ on strictly shorter expressions, the tree of recursive calls has at most $O(|E|)$ levels. Since $I(z,E)$ is always reduced to values of $I$ on disjoint expressions made up of subexpressions of $E$, each level of the tree can have at most $O(|E|)$ calls. Thus the entire tree has at most $O(|E|^2)$ calls with polynomial computation each, so $I(z,E)$ may be computed in polynomial time.
\end{proof}

We can now give an algorithm to solve the affine reachability problem over $\Z$ in full generality.

\begin{theorem} \label{main-algorithm}
There is a \textsf{2-EXPTIME} algorithm to decide, given any $x, y \in \Z$ and a finite set $F$ of functions $f_i(z) = a_i z + b_i$ with $a_i, b_i \in \Z$, whether $x \xrightarrow{F} y$.
\end{theorem}
\begin{proof}
There are several cases:
\begin{enumerate}
\item For some $j$, $a_j = 0$: Clearly, $x \xrightarrow{F} y$ if and only if either $x \xrightarrow{F \setminus \{f_j\}} y$ or $b_j \xrightarrow{F \setminus \{f_j\}} y$; recursively determine each of these and return true if and only if at least one is true. \label{ma-case-constant}
\item For some $j$, $a_j = 1$ and $b_j = 0$: Clearly $x \xrightarrow{F} y$ if and only if $x \xrightarrow{F \setminus \{f_j\}} y$ (since $f_j$ is the identity), so determine this recursively and return the result. \label{ma-case-identity}
\item For some $j$, $a_j = 1$ and $b_j \ne 0$: Assume for now that $b_j < 0$ (we will handle $b_j > 0$ momentarily). Run the algorithm of Lemma \ref{largest-mod-reachable} on $F \setminus \{f_j\}$ with $k = b_j$. If it returns \textsc{Empty}, then case \ref{case-none} of Lemma \ref{lemma-cases} holds, so return false. If it returns \textsc{Negative}, then case \ref{case-negative-coeff} holds, so return true. Otherwise, the algorithm returns $V = \sup \; \{ G(x) : x \xrightarrow{F \setminus \{f_j\}} y \pmod {b_j} \text{ via } G \}$. Since now either case \ref{case-opposite-sign} or case \ref{case-same-sign} of Lemma \ref{lemma-cases} holds, $x \xrightarrow{F} y$ if and only if $\sgn{V - y} \ne \sgn{b_j} = -1$. So we return true if and only if $V \ge y$. If $b_j$ was in fact positive, we use the variant of the algorithm of Lemma \ref{largest-mod-reachable} which computes $V' = \inf \; \{ G(x) : x \xrightarrow{F \setminus \{f_j\}} y \pmod {b_j} \text{ via } G \}$. By Lemma \ref{lemma-cases} again we have $x \xrightarrow{F} y$ if and only if $\sgn{V' - y} \ne \sgn{b_j} = +1$, so we return true if and only if $V' \le y$. \label{ma-case-shift}
\item For some $j$, $a_j = -1$, and $|a_i| > 1$ for all $i \ne j$: Use the algorithm of Lemma \ref{expanding-reachable}, modified as in the remark to handle $f_j(z) = -z + b_j$. \label{ma-case-involution}
\item For some $j, k$ with $j \ne k$, $a_j = a_k = -1$: Define $g = f_j \circ f_k$. Clearly $x \xrightarrow{F} y$ if and only if $x \xrightarrow{F \cup \{g\}} y$. But $g(z) = (f_j \circ f_k)(z) = -(-z+b_k)+b_j = z + (b_j - b_k)$, and since $b_j \ne b_k$ (since $f_j$ and $f_k$ are distinct functions), $g(z) = z + c$ for some $c \in \Z$ with $c \ne 0$. Recursively solve $x \xrightarrow{F \cup \{g\}} 0$ using case (\ref{ma-case-shift}) and return the result. \label{ma-case-two-involutions}
\item Otherwise, $|a_i| > 1$ for all $i$: Use the algorithm of Lemma \ref{expanding-reachable}. \label{ma-case-expanding}
\end{enumerate}
Cases (\ref{ma-case-involution}) and (\ref{ma-case-expanding}) invoke the algorithm of Lemma \ref{expanding-reachable} and take exponential time. Case (\ref{ma-case-shift}) invokes the algorithm of Lemma \ref{largest-mod-reachable} and takes doubly-exponential time. Case (\ref{ma-case-two-involutions}) makes a recursive call which always uses case (\ref{ma-case-shift}), so it also takes doubly-exponential time. Finally, cases (\ref{ma-case-identity}) and (\ref{ma-case-constant}) make one and two recursive calls respectively, each with one less affine function. Thus in total the algorithm will make at most exponentially-many recursive calls (this can easily be reduced to linearly-many by improving the handling of case (\ref{ma-case-constant}), but this does not decrease the worst-case runtime), with at most a doubly-exponential amount of computation each. Therefore the algorithm runs in at most doubly-exponential time.
\end{proof}

\section{Affine reachability over $\N$} \label{reach-n}

We can also consider the affine reachability problem over $\N$. Much of the analysis is the same, so we will only write out in full detail the considerations which are new. The main difference from the version over $\Z$ is that now we cannot apply any functions which would yield a negative result.
\begin{definition}
An $F$-composition is \emph{valid with respect to its argument} if every integer in its orbit for the given argument is nonnegative. Often the argument of the composition will be clear from context, in which case we will simply say that the composition is \emph{valid}. We write $x \xrightarrow{F}_+ y$ to indicate that there is a valid $F$-composition $G$ such that $G(x) = y$.
\end{definition}
With this definition, the affine reachability problem over $\N$ is to determine, given $x,y \in \N$ and a finite set $F$ of functions $f_i(z) = a_i z + b_i$ with $a_i,b_i \in \Z$, whether $x \xrightarrow{F}_+ y$. As before, there are various cases depending on the values of the linear coefficients $a_i$. The case where they all satisfy $|a_i| > 1$ is still simple.
\begin{lemma} \label{n-expanding-reachable}
There is an \textsf{EXPTIME} algorithm to decide, given any $x, y \in \N$ and a finite set $F$ of functions $f_i(z) = a_i z + b_i$ with $a_i, b_i \in \Z$ and satisfying $|a_i| > 1$, whether $x \xrightarrow{F}_+ y$.
\end{lemma}
\begin{proof}
Use the algorithm of Lemma \ref{expanding-reachable}, but with the interval $I=[0,R]$ instead of $[-R,R]$. The argument in the proof of Lemma \ref{expanding-reachable} goes through as before, since all preimages of $y$ under \emph{valid} $F$-compositions must lie in $I$.
\end{proof}
\begin{remark}
This algorithm also works with any number of functions of the form $g(z) = z + k$ with $k > 0$, since these strictly increase absolute value on $\N \setminus I$, and so preimages of $y$ under valid compositions including them must lie in $I$.
\end{remark}

The algorithm of Lemma \ref{n-expanding-reachable} cannot handle functions of the form $g(z) = z - k$ with $k > 0$. Fortunately, as before we can reduce problems with this type of function to ``modular'' problems without them, using the following (much simpler) analog of Lemma \ref{lemma-cases}. We assume for now that all $a_i > 0$, and show how to handle other cases later.

\begin{lemma} \label{n-lemma-cases}
For any $x,y \in \N$, a set $S$ of functions $f_i(z) = a_i z + b_i$ with $a_i, b_i \in \Z$ and $a_i > 0$ for $1 \le i \le N$, and function $g(z) = z - k$ with $k > 0$, put $F = S \cup \{g\}$. Then $x \xrightarrow{F}_+ y$ if and only if $x \xrightarrow{S}_+ z \equiv y \pmod k$ for some $z \ge y$.
\end{lemma}
\begin{proof}
Suppose $x \xrightarrow{F}_+ y$ via $G$. By Lemma \ref{extraction-lemma}\ref{total-extraction}, $G = g^a \circ H$ with $a \ge 0$ and $H$ being $G$ with all instances of $g$ removed. Since $G(x) = y$, we have $H(x) \equiv y \pmod k$. Now note that $g$ always decreases its argument, and since every $a_i$ is positive, each $f_i$ maps larger inputs to larger outputs. Therefore removing an instance of $g$ from $G$ can only increase the values of integers in its orbit. Since $G$ is valid this means $H$ must be as well, and since $G(x) = y$ this means $H(x) \ge y$. Therefore $x \xrightarrow{S}_+ H(x) \equiv y \pmod k$ with $H(x) \ge y$.

Conversely, suppose $x \xrightarrow{S}_+ z \equiv y \pmod k$ via $H$ with $z \ge y$. Then $(g^n \circ H)(x) = z - nk = y$ for some $n \in \N$. Putting $G = g^n \circ H$, $G$ is valid because $H$ is, and so we have $x \xrightarrow{F}_+ y$ via $G$.
\end{proof}

The algorithm of Lemma \ref{largest-mod-reachable} is almost exactly what we need to test the condition in Lemma \ref{n-lemma-cases}, since it computes the largest $z \equiv y \pmod k$ reachable by $S$-compositions. However, we must consider only those reachable by \emph{valid} $S$-compositions, and so need to modify the algorithm. This is not hard to do.

\begin{lemma} \label{n-largest-mod-reachable}
Given any $x, y \in \N$, $k \in \Z$ with $k \ne 0$, and a set $F$ of functions $f_i(z) = a_i z + b_i$ with $a_i, b_i \in \Z$ and $a_i > 0$ for $1 \le i \le N$, put $\mathbf{G} = \{ G : x \xrightarrow{F}_+ y \pmod k \: \text{via} \; G \}$. There is a \textsf{2-EXPTIME} algorithm which returns \textsc{Empty} if $\mathbf{G} = \emptyset$, and otherwise returns $\sup \; \{ G(x) : G \in \mathbf{G} \}$.
\end{lemma}
\begin{proof}
As in the algorithm of Lemma \ref{largest-mod-reachable}, construct the graph $D$ and search it to determine if $x \xrightarrow{F} y \pmod k$ via any $F$-composition, not necessarily a valid one. If not, return \textsc{Empty}. Otherwise, consider $D$ to be a finite automaton as before, and convert it into a reduced regular expression $R$ in disjunctive normal form $R = S_1 | \dots | S_M$.

Given some $z \in \N$ and $\ell$ a sequence of literals which appear in $R$, we define $V(z,\ell)$ to mean that $P_\ell$ is a valid $F$-composition with respect to $z$. Now for any $z \in \N$ and reduced expression $E$, define $I'(z,E)$ to mean $\exists s \in L(E) : V(z, P_s) \land (P_s(z) > z)$. In words, this means that there is some $F$-composition valid with respect to $z$ which matches $E$ and increases $z$ (this is just the analog of $I(z,E)$ from Lemma \ref{largest-mod-reachable}, but restricted to only valid compositions). By an extension of Lemma \ref{I-lemma} which will prove momentarily, Lemma \ref{n-I-lemma}, we have that $I'(z, \ell \alpha^* \beta) \iff V(z, \ell) \land (I'(P_\ell(z), \alpha) \lor I'(z, \ell \beta))$. Using this we may compute $I'(z,E)$ in polynomial time with the analog of the algorithm of Lemma \ref{I-algorithm}, described shortly in Lemma \ref{n-I-algorithm}.

Now we continue as in the algorithm of Lemma \ref{largest-mod-reachable}, writing $S_i = T_1 \dots T_K$ and defining $x^{(i)}_j = \sup \; \{ P_s(x) : s \in L(T_1 \dots T_j) \}$. We calculate the values $x^{(i)}_j$ in the same way as before, except using $I'$ in place of $I$ when dealing with starred subexpressions $T_j$. This ensures that only $F$-compositions which are valid with respect to $x^{(i)}_{j-1}$ are used to compute $x^{(i)}_j$ when $T_j$ is a starred subexpression. When $T_j$ is a literal, we put $x^{(i)}_j = P_{T_j}(x^{(i)}_{j-1})$ as usual, but also check if $x^{(i)}_j < 0$. If so, then $P_{T_j}$ is not valid with respect to $x^{(i)}_{j-1}$, and thus no $F$-composition matching $S_i$ can be valid with respect to $x$. So we discard $S_i$ and move on. Otherwise again $x^{(i)}_j$ can be obtained from $x^{(i)}_{j-1}$ using a valid $F$-composition. Then if we compute $x^{(i)}_K$ without discarding $S_i$, the value $x^{(i)}_K$ can be obtained from $x$ using a valid $F$-composition, and so $V_i = x^{(i)}_K$ is the supremum of possible values $x$ is mapped to by any valid $F$-composition matching $S_i$.

If we discarded every $S_i$, then no valid $F$-compositions match $R$, so $\mathbf{G} = \emptyset$ and we return \textsc{Empty}. Otherwise, if $V$ is the largest of the values $V_i$ (at least one of these is defined since we did not discard every $S_i$) then $V = \sup \; \{ G(x) : G \in \mathbf{G} \}$ and we return it. As in Lemma \ref{largest-mod-reachable}, this algorithm takes time polynomial in $|R|$, and thus takes at most doubly-exponential time.
\end{proof}

Now we prove the analog of Lemma \ref{I-lemma} for $I'$.

\begin{lemma} \label{n-I-lemma}
If $z \in \N$, $\alpha$ and $\beta$ are reduced regular expressions, and $\ell$ is a sequence of literals, the following are true:
\begin{enumerate}
\item $I'(z, \ell \alpha^* \beta) \iff V(z, \ell) \land (I'(z, \ell \beta) \lor I'(P_\ell(z), \alpha))$ \label{n-I-lemma-literal-star-prefix}
\item $I'(z, \ell \alpha^*) \iff V(z, \ell) \land (I'(z, \ell) \lor I'(P_\ell(z), \alpha))$ \label{n-I-lemma-literal-star}
\item $I'(z, \alpha^* \beta) \iff I'(z, \alpha) \lor I'(z, \beta)$ \label{n-I-lemma-star-prefix}
\item $I'(z, \alpha^*) \iff I'(z,\alpha)$ \label{n-I-lemma-star}
\end{enumerate}
\end{lemma}
\begin{proof}
\begin{enumerate}
\item $I'(z, \ell \beta)$ implies $I'(z, \ell \alpha^* \beta)$ because $\ell \beta$ matches $\ell \alpha^* \beta$. If $I'(P_\ell(z), \alpha)$ holds, then there is an $F$-composition $p$ matching $\alpha$ which is valid with respect to and increases $P_\ell(z)$. Because $p$ has a positive linear coefficient (since we assumed all $a_i > 0$), $p$ must increase anything greater than $P_\ell(z)$, and thus repeated applications of $p$ can increase $P_\ell(z)$ as much as desired. Therefore for any particular $F$-composition $q$ matching $\beta$, there is some $n \in \N$ such that $q \circ p^n$ is valid with respect to and increases $P_\ell(z)$, and $q \circ p^n \circ P_\ell$ increases $z$. If we also have $V(z, P_\ell)$, then $P_\ell$ is valid with respect to $z$, and then $q \circ p^n \circ P_\ell$ is valid with respect to and increases $z$. Then, since $q \circ p^n \circ P_\ell$ matches $\ell \alpha^* \beta$, $I'(z, \ell \alpha^* \beta)$ holds.

Suppose $V(z, \ell)$ does not hold. Then no $F$-composition $p$ matching $\ell \alpha^* \beta$ can be valid with respect to $z$, since the first part of $p$ must be $P_\ell$, which is not valid with respect to $z$ and thus will cause some integers in the orbit of $p$ to be negative. So $I'(z, \ell \alpha^* \beta)$ does not hold.

Suppose that neither $I'(z, \ell \beta)$ nor $I'(P_\ell(z), \alpha)$ hold. Any valid $F$-composition matching $\ell \alpha^* \beta$ is of the form $q \circ p \circ P_\ell$ where $q$ matches $\beta$ and $p$ is a composition of $F$-compositions matching $\alpha$. Since by our assumption valid $F$-compositions matching $\alpha$ do not increase $P_\ell(z)$, and these have positive linear coefficients, they do not increase anything smaller than $P_\ell(z)$. Thus no valid composition of $F$-compositions matching $\alpha$ increases $P_\ell(z)$, and so $p$ does not increase $P_\ell(z)$. Therefore because $q$ is valid with respect to $(p \circ P_\ell)(z)$, $q \circ P_\ell$ is valid with respect to $z$, and since $I'(z, \ell \beta)$ does not hold, $q \circ P_\ell$ does not increase $z$. Because $q$ has a positive linear coefficient, $q \circ p \circ P_\ell$ does not increase $z$ either. Thus no valid $F$-composition matching $\ell \alpha^* \beta$ increases $z$, and $I(z, \ell \alpha^* \beta)$ does not hold.
\item Put $\beta = \epsilon$ in (\ref{n-I-lemma-literal-star-prefix}) and use $L(\ell \alpha^* \epsilon) = L(\ell \alpha^*)$ and $L(\ell \epsilon) = L(\ell)$.
\item Put $\ell = \epsilon$ in (\ref{n-I-lemma-literal-star-prefix}) and use $L(\epsilon \beta) = L(\beta)$ and $P_\epsilon(z) = z$.
\item Put $\beta = \epsilon$ in (\ref{n-I-lemma-star-prefix}), use $L(\alpha^* \epsilon) = L(\alpha^*)$, and note that $I'(z, \epsilon)$ is clearly false. \qedhere
\end{enumerate}
\end{proof}
As before, these relationships yield a recursive algorithm for computing $I'$:
\begin{lemma} \label{n-I-algorithm}
There is a \textsf{P} algorithm to compute $I'(z, E)$ for any $z \in \N$ and reduced regular expression $E$.
\end{lemma}
\begin{proof}
The algorithm is identical to that of Lemma \ref{I-algorithm}, except that in addition to reducing $I'(z,E)$ to values of $I'$ on shorter expressions, the identities in Lemma \ref{n-I-lemma} also require evaluations of $V(z,\ell)$. Clearly $V(z,\ell)$ can be computed in polynomial time, by simply calculating the entire orbit of $P_\ell$ and checking that every integer in it is nonnegative. Adding a single computation of $V$ at each recursive call in the algorithm of Lemma \ref{I-algorithm} multiplies its runtime by a polynomial factor, so the overall algorithm still runs in polynomial time.
\end{proof}

Now we can give a general algorithm for the affine reachability problem over $\N$.

\begin{theorem}
There is a \textsf{2-EXPTIME} algorithm to decide, given any $x, y \in \N$ and a finite set $F$ of functions $f_i(z) = a_i z + b_i$ with $a_i, b_i \in \Z$, whether $x \xrightarrow{F}_+ y$.
\end{theorem}
\begin{proof}
There are several cases:
\begin{enumerate}
\item For some $j$, $a_j = 0$: As in Theorem \ref{main-algorithm}, recursively determine $x \xrightarrow{F \setminus \{f_j\}} y$ and $b_j \xrightarrow{F \setminus \{f_j\}} y$ and return true if and only if at least one holds.
\item For some $j$, $a_j < 0$: There are only a finite number of $z \in \N$ which $f_j$ can be applied to without giving a negative result: for example, they are all in $[0, b_j]$. Create a directed graph $D$ with a vertex for each of these, as well as vertices for $x$ and $y$ if not already present. Add edges indicating how $f_j$ maps these values to each other. Use this algorithm recursively on $S = F \setminus \{f_j\}$ to add edges corresponding to mappings by all possible valid $S$-compositions. Then there is a path in $D$ from $x$ to $y$ if and only if $x \xrightarrow{F}_+ y$. Use graph search to test if such a path exists, and return the result. \label{n-ma-negative}
\item For some $j$, $a_j = 1$ and $b_j = 0$: As in Theorem \ref{main-algorithm}, recursively solve $x \xrightarrow{F \setminus \{f_j\}} y$ and return the result.
\item For some $j$, $a_j = 1$ and $b_j < 0$: Run the algorithm in Lemma \ref{n-largest-mod-reachable} on $F \setminus \{f_j\}$ with $k = b_j$. If it returns \textsc{Empty}, then by Lemma \ref{n-lemma-cases} we cannot have $x \xrightarrow{F}_+ y$ and we return false. Otherwise the algorithm returns $V = \sup \; \{ G(x) : x \xrightarrow{F}_+ y \pmod k \text{ via } G \}$, and again by Lemma \ref{n-lemma-cases} we have $x \xrightarrow{F}_+ y$ if and only if $V \ge y$. So return true if and only if $V \ge y$.
\item Otherwise, for all $i$ we have $a_i \ge 1$, and if $a_i = 1$ then $b_i > 0$: Use the algorithm of Lemma \ref{n-expanding-reachable} (which works in this case as noted in the remark), and return the result.
\end{enumerate}
The analysis of the runtime is similar to that given in Theorem \ref{main-algorithm}, except for case \ref{n-ma-negative}. The graph $D$ created in that case has exponentially-many vertices (linear in the value of $b_j$), so each invocation of the algorithm makes at most exponentially-many recursive calls. There are at most a linear number of levels, since each recursive call has one less affine function than its parent. Thus there are exponentially-many recursive calls in total. The work done in each call takes at most doubly-exponential time (since the algorithm in Lemma \ref{n-largest-mod-reachable} can take this long), so the overall running time of the algorithm is at most doubly-exponential.
\end{proof}

\section{A Lower Bound} \label{lower-bound}

While it may be hoped that there are vastly more efficient algorithms for the affine reachability problems than the \textsf{2-EXPTIME} methods we have given here, the following theorem shows that polynomial-time algorithms are unlikely.

\begin{theorem}
The affine reachability problems over $\Z$ and $\N$ are \textsf{NP}-hard.
\end{theorem}
\begin{proof}
We give a reduction from the Integer Knapsack Problem (IKP), which is to determine, given $w_1, \dots, w_N, C \in \N$, whether there are $x_1, \dots, x_N \in \N$ such that $\sum_i w_i x_i = C$. This problem is known to be \textsf{NP}-complete \cite{lueker}. For a given instance of the IKP, $w_1, \dots, w_N, C \in \N$, let the set $F$ consist of the affine functions $f_i(z) = z + w_i$ for $1 \le i \le N$. If there exist $x_1, \dots, x_N \in \N$ such that $\sum_i w_i x_i = C$, then $(f_1^{x_1} \circ \dots \circ f_N^{x_N})(0) = \sum_i w_i x_i = C$, so $0 \xrightarrow{F} C$. Since the functions $f_i$ all commute, if $0 \xrightarrow{F} C$ then $C = (f_1^{x_1} \circ \dots \circ f_N^{x_N})(0) = \sum_i w_i x_i$ for some $x_1, \dots, x_N \in \N$. Thus $0 \xrightarrow{F} C$ if and only if the IKP instance is solvable. Computing $F$ given an IKP instance can obviously be done in polynomial time, so this gives a polynomial-time many-one reduction from IKP to the affine reachability problem over $\Z$, showing that the latter is \textsf{NP}-hard. The reduction to affine reachability over $\N$ is exactly the same, since all $F$-compositions are valid.
\end{proof}

\section{Conclusion}

We gave \textsf{2-EXPTIME} algorithms for the affine reachability problems over $\Z$ and $\N$, and showed that they are \textsf{NP}-hard. Beyond improving these upper and lower bounds, a natural generalization that would be interesting to consider is if integer or integer-valued polynomials are allowed instead of just affine functions. Also, the original problem which this paper treated a special case of, namely reachability for affine evolution over $\Q$, remains open. This provides another clear direction for future work.

\bibliographystyle{elsart-num}
\bibliography{main}

\end{document}